\documentclass[11pt]{article}
\usepackage{authblk}
\usepackage{amsmath,amssymb,fullpage,times,color,graphicx,subfigure}
\usepackage{algorithm, algorithmic}
\usepackage{tikz}
\usetikzlibrary{external}
\usetikzlibrary{arrows}
\newtheorem{theorem}{Theorem}[section]

\newtheorem{lemma}[theorem]{Lemma}

\newcommand{\qed}{\rule{2mm}{2mm}}
\newenvironment{proof}{\par\noindent{\bf Proof.}\quad}{  $\qed$}

\newcommand{\imply}{\Rightarrow}

\title{Firefighter Problem with Minimum Budget: Hardness and Approximation Algorithm for Unit Disk Graphs}	

\author{  
	Diptendu Chatterjee \thanks{
		Indian Statistical Institute,Kolkata, India.
		{\tt diptendu\_r@isical.ac.in}} 
	
\and  
	Rishiraj Bhattacharyya \thanks{
		NISER, Bhubaneswar, India.
		{\tt rishi@niser.ac.in}} 
}

\begin{document}
	
	\maketitle

	\begin{abstract}
		Unit disk graphs are the set of graphs which represent the intersection of disk graphs and interval graphs. These graphs are of great importance due to their structural similarity with wireless communication networks. Firefighter problem on unit disk graph is interesting as it models the virus spreading in an wireless network and asks for a solution to stop it. 
		In this paper, we consider the MIN-BUDGET firefighter problem where the goal is to determine the minimum number of firefighters required and the nodes to place them at each time instant to save a given set of vertices of a given graph and a fire breakout node. We show that, the MIN-BUDGET firefighter problem in a unit disk graph is NP-Hard. We also present a constant factor approximation algorithm.
		
	\end{abstract}

	\section{Introduction} $\ \ \ $
	\textsc{Firefighter} Problem on a graph is broadly described by saving vertices of a graph from fire by placing fire fighters on them, given a source vertex of fire and fire spreading to every neighbour of a burning vertex unless fire fighters are placed on them.  Firefighter Problem can be classified in two main categories depending on the set of vertices to be saved by it, i.e. \textsc{Max-Save} and \textsc{Min-Budget}. In the \textsc{ Max-Save} version, the goal is to save maximum number of vertices possible. In contrast, in the \textsc{ Min-Budget} version, the goal is to save a given set of vertices using minimum number of fire fighters at each time step.
	
	Firefighter Problem has very important relevance to many real life problems. Suppose a virus is spreading through a network and our goal is to prevent some nodes in the network from this virus by using some preventive measure at some nodes, then this can be seen as a firefighter problem on this network where the virus is like the spreading fire and the preventive measure can be thought as firefighters. In another case, suppose a sensitive data or information has been leaked in a communication network and we need to restrict the flow of the information to some vulnerable or less secured nodes. This scenario also can be modeled as a firefighter problem, where the information can be seen as the spreading fire and the vulnerable nodes as the vertices to be saved using firefighters. 
	
	\textsc{Unit Disk} Graph represents the intersection of disk graph and interval graph. For each vertex $u$ in the graph, there is a disk of radius $1$ in the Euclidean plane. There is an edge between vertex $u$ and $v$, if $v$ is inside the unit disk centered at $u$. 
	
	Unit Disk Graphs are relevant to communication network as in almost all practical communication networks the towers or routers have same communicating range which can be modeled as disks of coverage area similar to unit disk graph and to prevent any virus or sensitive information from reaching a node we have to take preventive measures which is same as using some firefighters at some nodes in the firefighter problem case.  In most of the practical scenario in a communication network the goal is to keep the sensitive data away from some specific nodes in the network given the network topology. This case resembles Firefighter \textsc{Min-Budget Problem}, where we are given the source of fire and a set of vertices to be saved from fire by putting firefighters in some node at different time instances and minimize our cost while doing so. The cost here is the number of firefighters required at each time instant to achieve our goal.
	
	\subsection{Our Results}
	In this paper, we prove the following two results.
	
	\begin{theorem}
		The Firefighter MIN-BUDGET problem on unit disk graph is NP hard.
	\end{theorem}
	
	\begin{theorem}
		There exists an approximate solution to Fighter MIN-BUDGET problem on unit disk graph with a constant approximation factor.
	\end{theorem}
	
	\subsection{Previous Works}
	The Freighter problem has a rich history of work. We mention some results which are most relevant to our work. For a detailed account, we refer the reader to the survey on results, directions and open questions related to firefighter problem by \cite{finbow2009firefighter}. The firefighter problem was introduced by \cite{hartnell1995firefighter} who studied the way to protect an infinite graph from a spreading virus.  \cite{hartnell2000firefighting} showed that a simple 2-approximation algorithm exists for Firefighter Max-Save Problem. \cite{wang2002fire} showed 2 firefighters are required to control the fire in an finite 2 dimensional grid with minimum 8 steps and for (r-1) firefighters are required to save a r-regular graph . \cite{develin2007fire} further added that at least 18 vertices will burn to do so and (2d-1) firefighters can hold a fire in an n-dimensional grid if fire breaks out in a single vertex .\cite{macgillivray2003firefighter} proved NP-hardness of firefighter problem on Bipartite graphs and gave a polynomial time algorithm for the same problem on P-trees and also gave an integer programming to determine an optimal sequence of vaccinating the vertices . This work was followed by the work by \cite{hartke2004attempting} to narrow the integrality gap between the integer programming optimal and the optimal of the linear programming relaxation. \cite{finbow2007firefighter} showed that the firefighter problem is NP-hard on trees with maximum degree 3. However if fire breaks out at a vertex with degree 2, the problem is in P. \cite{king2010firefighter} showed that Resource Minimization for Fire Containment(RMFC) is NP-hard on full trees of degree 3. They also showed that the 2-approximation is  hard as well, but  on graphs of maximum degree 3, if fire breaks out at a vertex of degree 2, the problem is in P.  In a follow-up work of \cite{chalermsook2010resource}, the authors showed an $O(\log n)$-approximation LP-rounding algorithm for RMFC on trees. \cite{chlebikova2014firefighter} showed that, the complexity of firefighter problem is governed by path width and maximum degree of a graph. Work on parameterized complexity of firefighter problem can be found in the works by \cite{bazgan2011parameterized,cygan2011parameterized}. \cite{anshelevich2009approximation}  considered both spreading and non-spreading vaccination models of firefighter problem to show the NP-hardness of approximation of the problem within $n^{1-\epsilon}$ for any $\epsilon$.  \cite{fomin2012making,fomin2016firefighter} considered the  Max-Save Firefighter problem and gave polynomial time algorithm for interval graphs, split graphs, permutation graphs, and proved NP-hardness of the problem on Unit Disk Graphs. Online firefighter problem and fractional firefighter problem were introduced by \cite{coupechoux2019firefighting} and they also gave the optimal competitive ratio and asymptotic comparison of number of firefighters and the size of the tree levels. NP-completeness of different graph problems on unit disk graphs have been discussed by \cite{clark1990unit}.  
	
	\section{Notation and Preliminaries}
	We use the following notations. If $S$ is a set, $|S|$ denotes size of the set. $d_{\ell_2}(u,v)$ denotes the $\ell_2$ distance of two points $u$ and $v$. An instance of the firefighter problem is denoted by $(G,s)$ where $G$ is an undirected, connected, simple graph, and the fire starts from the vertex $s$. $V_G$ and $E_G$ denote the vertex set and the edge set of $G$ respectively.
	
	\subsubsection*{{\sc Min-Budget} Firefighter Problem }
	\label{sec:min-budg-firef}
	The {\sc Min-Budget} Firefighter Problem is defined in the following way.
	{\em Given a undirected, connected, simple graph $G$, a vertex $s\in V_G$, $T\subseteq V_G$, find the minimum number of firefighters needs to be placed at every step to save $T$.} 
	
	\subsubsection*{Decision Version of Min-Budget Firefighter Problem in Unit Disk Graph.}
	\begin{itemize}
		\item \textbf{Input}: A unit disk graph $G(V,E)$, a vertex $s\in V$, a subset $T\subseteq V$
		\item \textbf{Question}: If the fire starts from $s$, can all vertices of $T$ be saved if at most $B$ number of firefighter can be placed at each step?  
	\end{itemize}
	
	\subsubsection*{Interval Graph}
	\label{sec:interval-graph}
	Our approximation algorithm extends the exact algorithm for interval graphs. 
	An interval Graph is an undirected graph formed by correspondence with a set of intervals on real line, where each interval corresponds to a vertex in the graph and two vertices have an edge in between them if their corresponding intervals intersect. 
	
	\subsubsection*{Unit Disk Graph}
	\label{sec:unit-disk-graph}
	Unit Disk Graph is an undirected graph formed by correspondence with a set of disks of unit radius on euclidean plane, where each disk correspond to a vertex in the graph and two vertices have edge in between them if their corresponding disks intersect.
	
	\section{Proof of theorem 1.1:  the Reduction}
	\label{sec:proof-hard}
	We show a reduction from the Firefighter Problem for Trees with maximum degree three which has been shown to be NP hard. The proof idea is similar to the one of the MAX-SAVE version due to \cite{fomin2016firefighter}.  The first step is to transform an instance of the MIN-BUDGET Firefighter problem on Trees with maximum degree Three to an instance of the MIN-BUDGET firefighter problem on an Unit disk graph. In order to do that we show an embedding of any full rooted degree three tree into an unit disk graph.
	
		
	
	\subsection{The Construction}
	\label{sec:construction}
	Let $(T,\Gamma,B)$ be the input  instance of the firefighter problem on Trees with maximum degree three. $\Gamma$ denotes the set of vertices to be saved and $B$ denotes the budget. The first step is embed $T$ into a unit disk graph $T_G$.
	\textsc{Rectilinear Embedding.}   Let $T$ be a full rooted tree of maximum degree three and $r$ be the root of $T$. Let $m$ be the number of vertices in $T$. We add a new vertex $s$ in the tree and connect it to $r$. Let us call the new graph $T'$ rooted at $s$ with $N=m+1$ vertices. 
	
	We number the leaves of the new graph $T'$ according to their appearance in the pre-order traversal starting from $s$. Then, for each non-leaf vertex, number it with the median of the numbers corresponding to its children. If any of them is having two children, then number it with the maximum number belonging to its children. We denote the number corresponding to the vertex $v$ after numbering in this fashion by $n(v)$.
	
	The embedding works in the following way. Put the vertex $s$ at the coordinate $(2n(s), 0)$. Let $v$ be any other vertex. Suppose the parent of $v$, say $u$, has been placed at $(X_u,Y_u)$. We place $v$ at $(X_v,Y_v)$, where $X_v=X_u+ 2\left(n(u)-n(v)\right)$, and $Y_v= Y_u+ 2 \left(N- \left|n(u)-n(v)\right|\right)$.
	
	The edges are mapped to paths parallel to axis. We draw the edge between $u$ and $v$ as follows. Draw a line of length $\left(N- \left|n(u)-n(v)\right|\right)$ in the positive $Y$ direction from $u$ and then a line of length $2|n(u) - n(v)|$ in the positive/negative direction of $X$ axis, according to the positive/negative values of $n(u) - n(v)$. By construction this embedding is indeed a rectilinear embedding. Moreover, in this embedding each edge is of length exactly $2N$ and has at most one bend.
	
	The final step of our transformation is the following. On each edge starting from a parent $u$ to a child $v$ in $T'$, we put $2N - 1$ new vertices, one at every point in the path joining $u$ to $v$.  $W_v^1$ to $W_v^{2N-1}$ denotes these new vertices, where $W_v^1$ is adjacent to $u$ and $W_v^{2N-1}$ is adjacent to $v$. Finally, for each vertex $v$ in $T'$, we split $W_v^1$, $2N - 1$ times and split $W_v^2,W_v^3,....,W_v^{2N-1},v$, $4N - 1$ times. We call this graph $T_G$.
	
	\begin{lemma}
		\label{lemma:const}
		$T_G$ is a unit disk graph.
	\end{lemma}
	
	\begin{proof}
		Consider any edge $(w_1,w_2)$ in $T_G$. By construction, $w_1,w_2$ are two adjacent vertices in the path joining vertex $u$ and vertex $v$ of $T$. As all the paths are drawn parallel to one of the axis, $d_{\ell_2}(w_1,w_2)=1$.\\
		For the other direction, consider any non-adjacent pair of vertices $w_1,w_2$ of $T_G$. If the vertices are in the same path (joining two vertices of $T$ in the embedding), then there is at least one vertex between $w_1$ and $w_2$ in the path. As the vertices are placed at distance $1$, and each path has only one bend, $d_{\ell_2}(w_1,w_2)\geq \sqrt{2}>1$.
		
		Suppose, $w_1$ and $w_2$ are in different paths. By construction, distance between two completely different paths in $T_G$ is at least $2$. If some part of the paths are same, then  $d_{\ell_2}(w_1,w_2)\geq \sqrt{2}>1$. Hence, all the non adjacent vertices in $T_G$ are more than unit distance away from each other.
	\end{proof}
	
		
	
	The last step of the above construction produces $4N$ images of every vertex $v\in T$.  Let $\Gamma_G$ be the set of all $4n$ images for each of the vertices $\Gamma$.
	The constructed instance of the firefighter problem is $(T_G,\Gamma_G,B)$. The following lemma proves the correctness of the reduction.
	
	\begin{lemma}
		\label{lemma:red}
		$(T_G ,\Gamma_G)$ is a YES instance for the Decision version of the Minimum Budget Firefighter Problem on unit disk graph, if and only if $(T,\Gamma)$ is a yes instance of the Minimum Budget Firefighter problem on tree.
	\end{lemma}
	
	\subsubsection{Proof of Lemma \ref{lemma:red}}
	\label{sec:proof-lemma-refl}
	
	The proof of Lemma \ref{lemma:red} follows from following two lemmas.
	
	\begin{lemma}
		There is a strategy to save all $4N$ images of some vertex $v$ of $T_G$ in firefighter game if and only if there is a strategy to save all $2N$ images of $W_v^1$ of $T_G$.
	\end{lemma}
	
	\begin{proof}
		If in the graph $T_G$ all the $4N$ images of some vertex $v$ is saved then it means at least one of the vertices $W_v^2,W_v^3,....,W_v^{2N-1},v$  had all its $4N$ images vaccinated till fire came to that level or all the $2N$ images of $W_v^1$ were vaccinated before fire came to it. Now, if the first case is true for $v$ or some $W_v^k$ where $k \in \{2, 3, . . . . . . , 2N - 1\}$ then at least $2N$ images of $W_v^k$ were vaccinated when fire reached to $W_v^1$. Then, in spite of vaccinating $2N$ images of $W_v^k$, we could have vaccinated all the images of $W_v^1$. Also in a contra-positive sense if by any optimal strategy we can't save all the $4N$ images of $v$ before it reaches that level, then at least one vertex of $W_v^1$ must be infected when fire reached its level. So in $T_G$, saving $v$ means saving $W_v^1$ and vice-versa. To make this argument valid for $W_t^1$, where $t$ is the child of $r$ in $T$, we added the extra vertex $s$ to $T$ in the construction.
	\end{proof}

	\begin{lemma}
		There is a strategy to save $W_v^1$ where $W_u^1$ can't be saved in $T_G$ in firefighter game, if and only if there is a strategy to save $v$ where $u$ can't be saved in $T$ given $v$ is the child of $u$ in $T$.
	\end{lemma}
	
	\begin{proof}
		Now to correspond unit disk graph problem to the case of full rooted tree of maximum degree three problem, it is sufficient to prove that
		
		\begin{enumerate}
			\item \label{en:1} If it is possible to save vertex $v$ while its parent $u$ can't be saved in a full rooted tree then it is possible to save $W_v^1$ when $W_u^1$ can't be saved in $T_G$.
			\item \label{en:2} In a contra-positive sense if there is no strategy to save $v$ while its parent $u$ is infected in $T$ then there is no strategy in $T_G$ to save $W_v^1$ when $W_u^1$ is infected and fire came ahead of it.
			
		\end{enumerate}
		
		Now, for case \ref{en:1}, when $u$ just gets infected and the fire is about to come towards $v$, we vaccinate $v$. The similar case in $T_G$ will be when fire reaches $W_u^1$ and affects some of its images, we start vaccinating the images of $W_v^1$ and it takes $2N$ time steps to reach $W_v^1$ by then we vaccinate all the images of it and save $W_v^1$.

		For case \ref{en:2}, when $u$ is infected and fire started to move towards $v$ and by no strategy we can vaccinate $v$, the similar case in $T_G$ will be when $W_u^1$ is already infected and fire moved towards $u$ and affected $W_u^2$ and none of the images of $W_v^1$ is vaccinated. Then by the time fire reaches $W_v^1$, we can't save all the images of it as it will take $2N - 1$ time steps and we have to vaccinate $2N$ vertices. So, $W_v^1$ can't be saved.
	\end{proof}
	
	In \textbf{Lemma 3}, the equivalence in terms of burning or saving between vertices $v$ and $W^1_v$ of graph $T_G$ has been shown and in \textbf{Lemma 4}, the equivalence of saving strategies of vertex $W^1_v$ in graph $T_G$ and vertex $v$ in graph $T$ has been shown. These two lemmas together prove that saving or burning of vertex $v$ in graphs $T_G$ and $T$ are equivalent and thus prove \textbf{Lemma 2}.
	
	\textbf{Lemma 1} and \textbf{Lemma 2} together proves \textbf{Theorem 1}. 
	
	\section{Approximation Algorithm for MIN-BUDGET problem} $\ \ \ $
	Firefighter problem on interval graphs are polynomial time solvable. Keeping this fact in mind, the underlying geometry of unit disk graph in Euclidean plane can be modified a bit to get another geometry which has a polynomial time solution.
	
	Following is a polynomial time algorithm for the firefighter problem on interval graphs, which will be extended later for an approximation algorithm for the same problem on unit disk graph.

	\subsection{Exact Algorithm for Interval Graph}
	\label{sec:min-budget-problem}
	
	Let $(G,s,T)$ be an instance of the MIN-BUDGET Firefighter problem where $s$ is the source vertex  and $T \subset V_G$ is the set of vertices to be saved. Algorithm 1 formally describes algorithm.

	Algorithm \ref{alg:interval} is an exact algorithm.
	
	\begin{algorithm}[htb]
		\caption{Algorithm for Firefighter Min-Budget Problem on Interval Graph}
		\label{alg:interval} 
		\begin{algorithmic}
			\STATE \textbf{INPUT}: $G(V,E),v,T|$ where $G$ is an Interval Graph, $v \in V$ be the source vertex of fire and $T \subseteq V$ be the set of vertices to save.
			\STATE $T_{1}=T \cap N(v)$
			\STATE $T'=T \setminus T_{1}$
			\STATE $V_{1}=V$
			\STATE $A_{1}=|n(v)|$
			\STATE $B_{1}=|T_{1}|$
			\FOR{$i=2:n:1$}
			\IF{$|T'| \neq 0$}
			\STATE $V_{i}=V_{i-1} \setminus T_{i-1}$.
			\STATE $G_{i}(V_{i},E_{i})| \forall u,v \in V_{i}, (u,v) \in E \imply (u,v) \in E_{i}$
			\STATE $T_{i}=\{t \in T'| Dist_{G_{i}}(v,t)=i\}$
			\STATE $K_{i}=$ Minimum Cut-set of $T_{i}$ in $G_{i}$
			\STATE $B_{i}=\Bigl\lceil\frac{\sum_{k=1}^{i}|T_{k}|}{i}\Bigr\rceil$
			\STATE $A_{i}=\Bigl\lceil\frac{\sum_{k=1}^{i-1}|T_{k}|+|K_{i}|}{i}\Bigr\rceil$
			\ENDIF
			\ENDFOR
			
			\IF{$\max_{i}B_{i} \leq \min_{i}A_{i}$}
			\STATE Then Budget is $\max_{i}B_{i}$ and firefighters are to be placed on $T$ \ELSE \STATE Let $\min_{i}A_{i}=A_{m}$, Then Budget is $\max(A_{m},\max_{k}B_{k}|k < m)$ and firefighters are to be placed on $\cup_{i=1}^{m-1} T_{i} \cup K_{m}$
			
			\ENDIF
		\end{algorithmic}
		
	\end{algorithm}

	\paragraph{Correctness of Algorithm \ref{alg:interval}} $\ $
	
	\begin{lemma}
		In the optimal strategy, Firefighters have to be placed on $T_{i}$, if not placed on all the vertices of $K_{i}$ at $i^{th}$ time instant. 
	\end{lemma}
	
	\begin{proof}
		Consider the $i^{th}$ time instance.  Suppose, there are some unprotected vertices of $T_i$. Moreover suppose there are some vertices in $K_{i}$ where no firefighter is placed.  Then $\exists u_{i} \in T_{i}$, where fire reaches to the unprotected vertices of $T_{i}$ at $i^{th}$ or $(i+1)^{th}$ time instant, which is not permitted. So, we have to place firefighters on the unprotected vertices of $T_{i}$ at $i^{th}$ or $(i+1)^{th}$ time instant. So, placing the firefighters on $K_{i} \setminus T_{i}$ at $i^{th}$ time instant will be waste of firefighters and will also increase the budget. So, it will be wise to place the firefighters either on $K_{i}$ or on $T_{i}$ at $i^{th}$ time instant.
	\end{proof}
	
	\begin{lemma}
		If at $i^{th}$ time instant firefighters are placed on $K_{i}$ then fire will never reach to $T_{k} \forall k \geq i$.
	\end{lemma}
	
	\begin{proof}
		Suppose at $i^{th}$ time instant firefighters are placed on all vertices of $K_{i}$, then fire can never reach to any of the vertices of $W_{i}|W_{i}=T_{i} \setminus K_{i}$. The span of $W_{i} \geq 1$ as the intervals are of unit length. So, fire will never reach to $T_{k} \forall k \geq i$.
	\end{proof}
	
	\paragraph{}
	If at some level $i$, the better strategy is to put firefighters on $K_{i}$ then the fire fighting process will stop at that time instant. So, either we have to place fire fighter on $K_{i}$ at $i^{th}$ time instant and on $T_{k}$'s at corresponding time instant for $k < i$ or we have to put fire fighters on $T_{i}$'s $\forall i$. Now, if the optimum strategy is to place fire fighters on $K_{i}$ at $i^{th}$ step then the budget must be greater than both $A_{m}$ and $\max_{k}B_{k}|k < m$, where $\min_{i}A_{i}=A_{m}$ as both are the necessary conditions. But, if $\max_{i}B_{i} \leq \min_{i}A_{i}$, then $\max_{i}B_{i}$ gives the optimum budget as a necessary condition. This proves the correctness of the algorithm.

	\subsection{Approximation Algorithm for Unit Disk Graph}
	\label{sec:appr-algor-unit}
	
	\begin{lemma}
		Firefighter MIN-BUDGET problem on Unit Disk Graph has an approximate solution using the algorithm for MIN-BUDGET problem on Unit Interval Graph.
	\end{lemma}
	
	\begin{proof}
		
		Let, there be $n$ vertices in the unit disk graph $G(V,E)$. \textit{v} be the source vertex of fire and $T \subset V$ be the set of vertices to be saved. We first replace the disks with their circumscribing squares with sides parallel to the coordinate axes. Now, this problem can be seen as four firefighter MIN-BUDGET problems on Unit interval graphs by observing that fire can spread in the up, down, left or right direction accordingly starting from \textit{v}. So, we can use our algorithm for MIN-BUDGET problem on Unit Interval Graph simultaneously in the four direction around \textit{v}. Due to this modification in this solution an approximation factor will be introduced in the result. For a square(say $S_u$) in the Euclidean plane corresponding to a vertex(say $u$) of the graph, the leftmost boundary of $S_1^u$ is given by $x=X_l^u$ and the down most boundary is given by $y=Y_d^u$. Algorithm for solving firefighter MIN-BUDGET problem on this modified graph is given in Algorithm \ref{alg:unitdisk}.
		
		\begin{lemma}
			Due to approximating the firefighter MIN-BUDGET problem on Unit Disk Graph as combination of four firefighter MIN-BUDGET problem on Unit Interval Graph the solution introduces a constant approximation factor of $2$.
		\end{lemma}
		
		\begin{proof}
			Firstly, we eliminate the effect of modifying the disks into circumscribing squares by checking at each time step of the spread of fire whether the new disks are having any intersection with the disks defining the boundaries of the burning rectangle. In the algorithm a forced assumption is that the region in fire at any instant of time forms a rectangle in the Euclidean plane in not always true.Due to the burning rectangle approximation, a vertex corresponding to an unit disk centered at $(X,Y)$ can burn as early as at Max$(X,Y)$ time instant, while in actual case it can take at most $2(X+Y)$ time. As the fire has to travel at most $2k$ squares to reach $(k,0)$ point starting from $(0,0)$ (limiting case when every square has an neighbor square slightly shifted in $X$ direction). Due to this effect the partial average at each time instant can be as big as $2(X+Y)/$max$(X,Y)$ times of the actual. Here the approximation factor becomes $2(X+Y)/$ max$(X,Y)$, which is less than $2$.
		\end{proof}
		
		\textbf{Lemma 7} and \textbf{Lemma 8} together proves \textbf{Theorem 2}.
		
		\begin{algorithm}[]
			\small
			\caption{Firefighter Min-Budget Problem on Unit Disk Graph}
			\label{alg:unitdisk}
			\begin{algorithmic}[1]
				\STATE \textbf{INPUT}: $G(V,E),v,T|$ where $G$ is an Unit Disk Graph, $v \in V$ be the source vertex of fire and $T \subseteq V$ be the set of vertices to save.
				\STATE Replace all the disks by their circumscribing squares with sides parallel to coordinate axes and $G(V,E)$ becomes an Unit Square Graph.
				\STATE $\forall v_{i} \in V$ sort them according to the $X_{l}^{i}$ values of their corresponding Square $S_{i}$ in increasing order. If $X_{l}^{i}=X_{l}^{j}$ then sort them according to their increasing order of $Y_{d}^{i}$ and $Y_{d}^{j}$.
				\STATE $\forall v_{i} \in V$ sort them according to the $Y_{d}^{i}$ values of their corresponding Square $S_{i}$ in increasing order. If $Y_{d}^{i}=Y_{d}^{j}$ then sort them according to their increasing order of $X_{l}^{i}$ and $X_{l}^{j}$.
				\STATE The boundaries of the burning rectangle ($R$) be, $R_l^0=X_l^v, R_r^0=X_r^v, R_u^0=Y_u^v, R_d^0=Y_d^v$.
				\STATE $T^0=\phi$.
				\STATE $T'=T \setminus T^0$.
				\STATE $V_{0}=V$.
				\STATE $C_m^0=\{v\} \forall m \in \{l,r,u,d\}$.
				\FOR{$i=1:n:1$}
				\IF{$T' \neq \phi$}
				\STATE $V_r^i=W\ |\ \forall w \in W \subset V, R_r^{i-1}-1<S_l^w<R_r^{i-1} \ \& \  R_d^{i-1}-1<S_d^w<R_u^{i-1} \ \& \ \exists v \in C_r^{i-1}|(S_l^w-S_l^v)^2+(S_d^w-S_d^v)^2 \leq 1$.
				\STATE $V_l^i=W\ |\ \forall w \in W \subset V, R_l^{i-1}+1>S_r^w>R_l^{i-1} \ \& \  R_d^{i-1}-1<S_d^w<R_u^{i-1} \ \& \ \exists v \in C_l^{i-1}|(S_r^w-S_r^v)^2+(S_d^w-S_d^v)^2 \leq 1$.
				\STATE $V_u^i=W\ |\ \forall w \in W \subset V, R_u^{i-1}>S_d^w>R_u^{i-1}-1 \ \& \  R_l^{i-1}-1<S_l^w<R_r^{i-1} \ \& \ \exists v \in C_u^{i-1}|(S_l^w-S_l^v)^2+(S_d^w-S_d^v)^2 \leq 1$.
				\STATE $V_d^i=W\ |\ \forall w \in W \subset V, R_d^{i-1}<S_u^w<R_d^{i-1}+1 \ \& \  R_l^{i-1}-1<S_l^w<R_r^{i-1} \ \& \ \exists v \in C_d^{i-1}|(S_l^w-S_l^v)^2+(S_u^w-S_u^v)^2 \leq 1$.
				\FOR {$m \in \{l,r,u,d\}$}
				\STATE $T_m^i=T \cap V_m^i$.
				\IF{$V_m^i \setminus T_m^i \neq \phi$}
				\STATE $C_m^i=V_m^i \setminus T_m^i$
				\ELSE 
				\STATE $C_m^i=C_m^{i-1}$
				\ENDIF
				\STATE $R_m^i=X_m^v||X_m^v|=\max_{w}|X_m^w|\forall w \in C_m^i$
				\ENDFOR
				\STATE $T^i=T_r^i \cup T_l^i \cup T_u^i \cup T_d^i$.
				\STATE $T'=T'\setminus T^i$.
				\STATE $V_{i}=V_{i-1} \setminus T^i$.
				\STATE $G_{i}(V_{i},E_{i})| \forall u,v \in V_{i}, (u,v) \in E \imply (u,v) \in E_{i}$.
				\STATE $K_m^i=$ Minimum Cut-set of $V_m^i$ in $G_i$ where, $m \in \{r,l,u,d\}$.
				\STATE $B_m^i=\Bigl\lceil\frac{\sum_{k=1}^{i}|T_m^k|}{i}\Bigr\rceil$ where, $m \in \{r,l,u,d\}$.
				\STATE $A_m^i=\Bigl\lceil\frac{\sum_{k=1}^{i-1}|T_m^k|+|K_m^i|}{i}\Bigr\rceil$ where, $m \in \{r,l,u,d\}$.
				\ENDIF
				\ENDFOR
				\STATE $\mathfrak{B}=0, \mathcal{V}=\phi$. 
				\FOR {each $m \in \{r,l,u,d\}$}
				\STATE $T_m= \cup_{i=1}^{n} T_m^i$.
				\STATE $\mathfrak{B}_m=0, \mathcal{V}_m=\phi$.
				\IF{$\max_{i}B_m^i \leq \min_{i}A_m^i$}
				\STATE Then Budget $\mathfrak{B}_m$ is $\max_{i}B_m^i$ and firefighters are to be placed on $\mathcal{V}_m=T_m$ \ELSE \STATE Let $\min_{i}A_m^i=A_m^t$, Then Budget $\mathfrak{B}_m$ is $\max(A_m^t,\max_{k}B_m^k|k < t)$ and firefighters are to be placed on $\mathcal{V}_m=\cup_{i=1}^{t-1} T_m^i \cup K_m^t$.
				\ENDIF
				\STATE $\mathfrak{B}=\mathfrak{B}+\mathfrak{B}_m$.
				\STATE $\mathcal{V}=\mathcal{V} \cup \mathcal{V}_m$.
				\ENDFOR 
				
			\end{algorithmic}
			
		\end{algorithm}
		
	\end{proof}
	
	\subsection{Analysis of Algorithm \ref{alg:unitdisk}}
	The algorithm for minimum budget firefighter problem on unit disk graph is based on the algorithm for minimum budget firefighter problem on unit interval graph described in Algorithm \ref{alg:interval}. We consider the problem on unit disk graph as a combination of four simultaneous interval graph problems in four different directions, i.e. left($l$), right($r$), up($u$), down($d$). For a given graph $G(V,E)$ and a fire break-out point $v \in V$, we have to find the minimum budget of firefighters to be placed at each time step to save a given set of vertices $T \subseteq V$.
	
	As a first step, we replace all the unit disks corresponding to the vertices of $G$ with their circumscribing squares with sides parallel to the rectilinear axes of the plane. Then we sort the vertices of $G$ in two different arrays one according to the X-coordinates of the centers of their corresponding squares and the other according to the Y-coordinates of the centers of their corresponding squares. These two sorted arrays are for use in different stages of the algorithm later.
	
	Afterwards, we define and initialize a number of variables namely Squares corresponding to the disk $v$ i.e. $(S^v)$and its boundaries$(S^v_m)$, boundaries of the burning rectangle$(R^i_m)$ at $i^{th}$ iteration and direction $m$, Set of vertices saved$(T^i_m)$ at $i^{th}$ iteration and direction $m$, vertices saved in $i^{th}$ iteration$(T^i)$ and in each direction$(T^i_m)$, vertices at the boundary of the burning rectangle at $i^{th}$ iteration in each direction$(C^i_m)$, Budget required at each direction$(\mathfrak{B})_m$ and set of vertices on which firefighters to be placed$(\mathcal{V}_m)$.
	
	Now at each iteration, we find the new set of vertices whose corresponding squares have inter section with burning rectangle and also filter them depending on whether their corresponding disks have intersection with any of the disks corresponding to the set of vertices constituting the boundary of the burning rectangle at previous iteration. From this we update the burning rectangle and its boundary, vertices to be saved at current iteration, set of vertices constituting the new boundary of the burning rectangle and the minimum cut set of the vertices that completely holds the fire at current iteration. Depending on the cut set and the saved vertices at each iteration, we fix the budget for each direction along with the final budget and the set of vertices to be placed firefighter on.

	\subsection{Correctness of the Algorithm}
	
	In the algorithm for Minimum Budget Firefighter Problem on Unit Disk Graph, the algorithm for Minimum Budget Firefighter Problem on Unit Interval Graph has been applied simultaneously in four directions namely, upwards, downwards, left and right using the subscript $m$. So the correctness of the algorithm in each of these four directions is guaranteed by the correctness of the algorithm for interval graph that is already proved in this paper(section 4.1). 
	
	It should also be noted that due to geometric assumption of burning rectangle an approximation factor of $2$ has been introduced by this algorithm as described in \textbf{Lemma 8}.

	\bibliographystyle{unsrt}
	\bibliography{Firefighter}
	
\end{document}